\documentclass[conference]{IEEEtran}

\usepackage[dvips]{color}
\usepackage{comment}
\usepackage{todonotes}
\usepackage{epsf}
\usepackage{epsfig}
\usepackage{times}
\usepackage{epsfig}
\usepackage{graphicx}
\usepackage{bbold}
\usepackage{mathtools}
\usepackage{mathrsfs}
\usepackage{amssymb}
\usepackage{pdfpages}
\usepackage{epstopdf}
\usepackage{dsfont}
\usepackage{lettrine} 
\usepackage{amsmath,epsfig,amssymb,algorithm,algpseudocode,amsthm,cite,url}
\usepackage{subcaption}
\allowdisplaybreaks
\usepackage{csquotes}
\usepackage{verbatim}
\usepackage[english]{babel}
\IEEEoverridecommandlockouts
\usepackage{amsmath,amssymb}

\usepackage[justification=centering]{caption}
\usepackage{verbatim}

\newtheorem{theorem}{\bf Theorem}

\begin{document}
	
\title{ Mobile Internet of Things: Can UAVs Provide an Energy-Efficient Mobile Architecture?}    
\author{\IEEEauthorblockN{Mohammad Mozaffari$^1$, Walid Saad$^1$, Mehdi Bennis$^2$, and M\'erouane Debbah$^3$}\vspace{-0.3cm}\\
	\IEEEauthorblockA{
		\small $^1$ Wireless@VT, Electrical and Computer Engineering Department, Virginia Tech, VA, USA, Emails: \url{{mmozaff , walids}@vt.edu} \\
		$^2$ CWC - Centre for Wireless Communications, Oulu, Finland, Email: \url{bennis@ee.oulu.fi}\\
		$^3$ Mathematical and Algorithmic Sciences Lab, Huawei France R \& D, Paris, France, and CentraleSup´elec,\\   Universit´e Paris-Saclay, Gif-sur-Yvette, France, Email: \url{merouane.debbah@huawei.com}
	}\vspace{-0.9cm}}
\maketitle\vspace{-0.7cm}
\begin{abstract}
In this paper, the optimal trajectory and deployment of multiple unmanned aerial vehicles (UAVs), used as aerial base stations to collect data from ground Internet of Things (IoT) devices, is investigated. In particular, to enable reliable uplink communications for IoT devices with a minimum energy consumption, a new approach for optimal mobility of the UAVs is proposed. First, given a fixed ground IoT network, the total transmit power of the devices is minimized by properly clustering the IoT devices with each cluster being served by one UAV. Next, to maintain energy-efficient communications in time-varying mobile IoT networks, the optimal trajectories of the UAVs are determined by exploiting the framework of optimal transport theory. Simulation results show that by using the proposed approach, 
 the total transmit power of IoT  devices for reliable uplink communications can be reduced by $56\%$ compared to the fixed Voronoi deployment method. Moreover, our results yield the optimal paths that will be used by UAVs to serve the mobile IoT devices with a minimum energy consumption.
 
\end{abstract} \vspace{-0.01cm}

\section{Introduction}\vspace{-0.1cm}

The use of unmanned aerial vehicles (UAVs) as flying wireless communication platforms has received significant attention recently  \cite{mozaffari2, HouraniModeling,Irem}. 
 On the one hand, UAVs can be used as wireless relays for improving connectivity of ground wireless devices and extending network coverage. On the other hand, UAVs can act as mobile aerial base stations to provide reliable downlink and uplink communications for ground users, and boost the capacity of wireless networks \cite{Irem} and \cite{Letter}. Compared to the terrestrial base stations, the advantage of using UAV-based aerial base stations is their ability to quickly and easily move. Furthermore, the high altitude of UAVs can enable line-of-sight (LoS) communication links to the ground users. 
  Due to the adjustable altitude and mobility, UAVs can move towards potential ground users and establish reliable connections with a low transmit power. Hence, they can provide a cost-effective and energy-efficient solution to collect data from ground mobile users that are spread around a geographical area with limited terrestrial infrastructure. 
 
 Indeed, UAVs can play a key role in the \emph{Internet of Things (IoT)} which is composed of small, battery-limited devices such as  
  sensors, and health monitors \cite{dawy,lien,Eragh}. These devices are typically unable to transmit over a long distance due to their energy constraints \cite{lien}. In such IoT scenarios, UAVs can dynamically move towards IoT devices, collect the IoT data, and transmit it to other devices which are out of the communication ranges of the transmitting IoT devices \cite{lien}. In this case, the UAVs play the role of moving aggregators for IoT networks. However, to effectively use UAVs for IoT communications, several challenges must be addressed such as optimal deployment and energy-efficient use of UAVs \cite{mozaffari2}.
 
 In \cite {mozaffari2}, we investigated the optimal deployment and movement of a single UAV for supporting downlink wireless communications. However, this work was restricted to a single UAV and focused on the downlink. The work in \cite{Han} analyzed the optimal trajectory of UAVs to enhance connectivity of ad-hoc networks. Nevertheless, this work did not study the optimal movement of multiple UAVs acting as aerial base stations. The work in \cite{ MozaffariTransport} studied the optimal deployment of UAVs and UAV-users association for static ground users with the goal of meeting users' rate requirements. In \cite{pang}, the authors used UAVs to efficiently collect data and recharge the clusters' head in a wireless sensor network which is partitioned into multiple clusters. However, this work is limited to a static sensor network, and does not investigate the optimal deployment of the UAVs. While the energy efficiency of uplink data transmission in a machine-to-machine (M2M) communication network was investigated in \cite{Nof} and \cite{Tu}, the presence of UAVs was not considered. 
  In fact, none of the prior studies \cite{mozaffari2, HouraniModeling,Irem}, and  \cite{Han, MozaffariTransport, Nof, pang, Tu}, addressed the problem of optimal deployment and mobility of UAVs for enabling reliable and  energy-efficient communications for mobile IoT devices.  


 
The main contribution of this paper is to propose a novel approach for deploying multiple, mobile UAVs for energy-efficient uplink data collection from mobile, ground IoT devices. First, to minimize the total transmit power of the IoT devices, we create multiple clusters where each one is served by one of the UAVs. Next, to guarantee energy-efficient communications for the IoT devices in mobile and time-varying networks, we determine the optimal paths of the UAVs by exploiting dynamic clustering and optimal transport theory \cite{villani2003}. Using the proposed method, the total transmit power of the IoT devices required for successful transmissions is minimized by the dynamic movement of the UAVs. In addition, the proposed approach will minimize the total energy needed for the UAVs to effectively move. The results show that, using our proposed framework, 
  the total transmit power of the devices during the uplink transmissions can be reduced by $56\%$ compared to the fixed Voronoi deployment method.  Furthermore, given the optimal trajectories for UAVs, they can serve the mobile IoT devices with a minimum energy consumption.
    
The rest of this paper is organized as follows. In Section II, we present the system model and problem formulation. Section III presents the optimal devices' clustering approach. In Section IV, we address the mobility of the UAVs using discrete transport theory. In Section V, we provide the simulation results, and Section VI draws some conclusions.\vspace{-0.05cm}

\section{System Model and Problem Formulation}

Consider an IoT system consisting of a set $\mathcal{L}=\{1,2,...,L\}$ of $L$ IoT devices deployed within a geographical area. In this system, a set $\mathcal{K}=\{1,2,...,K\}$ of $K$ UAVs must be deployed to collect the data from the ground devices in the uplink. The locations of device $i\in \mathcal{L}$ and UAV $j\in \mathcal{K}$ are, respectively, given by $(x_i,y_i)$ and $(x_{u,j},y_{u,j},h_j)$ as shown in Figure 1. We assume that devices transmit in the uplink using orthogonal frequency division multiple access (OFDMA) and UAV $j$ can support at most $M_j$ devices simultaneously. Note that, we consider a centralized network in which the locations of devices and UAVs are  known to a control center such as a central cloud server. The ground IoT devices can be mobile (e.g. smart cars) and their data availability can be intermittent (e.g. sensors). Therefore, to effectively respond to the network mobility, it is essential that  the UAVs optimally move for  establishing reliable and energy-efficient communications with the devices. Note that, in our analysis, without loss of generality any mobility model can be accommodated.

 For  ground-to-air communications, each device will typically have a LoS view towards a specific UAV with a given probability. This LoS probability depends on the environment, location of the device and the UAV, and the elevation angle between the device and the UAV \cite{HouraniModeling}. One suitable expression for the LoS probability is given by \cite{HouraniModeling}:\vspace{-0.12cm}
\begin{equation}\label{PLoS}
	{P_{{\rm{LoS}}}} = \frac{1}{{1 + \psi \exp ( - \beta\left[ {\theta  - \psi} \right])}},
\end{equation}
where $\psi$  and $\beta$  are constant values which depend on the carrier frequency and type of environment such as rural, urban, or dense urban, and $\theta$  is the elevation angle.  Clearly, ${\theta} = \frac{{180}}{\pi } \times {\sin ^{ - 1}}\left( {{\textstyle{{{h_j}} \over { {d_{ij}}}}}} \right)$, where $ {d_{ij}} = \sqrt {(x_i-x_{u,j})^2+(y_i-y_{u,j})^2+h_j^2 }$ is the distance between device $i$ and UAV $j$. 

 From (\ref{PLoS}), we can observe that by increasing the elevation angle or increasing the UAV altitude, the LoS probability increases. We assume that, the necessary condition for connecting a device to a UAV is to have a LoS probability greater than a threshold ($\epsilon$ close to 1). In other words, ${P_{\text{LoS}}}(\theta ) \ge \varepsilon$, and hence, $\theta \ge P_{\text{LoS}}^{ - 1}(\varepsilon )$ leading to:\vspace{-0.2cm}  
 \begin{equation}\label{dmin}
 d_{ij} \le \frac{h_j}{{\sin \left( {P_{\text{LoS}}^{ - 1}(\varepsilon )} \right)}}. 
\end{equation}
Note that (\ref{dmin}) shows the necessary condition for connecting the device to the UAV. Therefore, a device will not be assigned to UAVs which are located at distances greater than $\frac{h_j}{{\sin \left( {P_{\text{LoS}}^{ - 1}(\varepsilon )} \right)}}$.  

Now, considering the LoS link, the received signal power at UAV $j$ from device $i$ is given by \cite{HouraniModeling} (in dB):
 \begin{equation} \label{Pr}
P_r^{ij} = {P_{t,i}} - 10\alpha \log \left( {\frac{{4\pi {f_c}}}{c}{d_{ij}}} \right) - \eta,
\end{equation} 
where $P_{t,i}$ is the transmit power of device $i$ in dB, $f_c$ is the carrier frequency, $\alpha=2$ is the path loss exponent for LoS propagation, $\eta$ is an excessive path loss added to the free space propagation loss, and $c$ is the speed of light. 

In our model, the transmit power of the devices must satisfy the minimum signal-to-noise-ratio (SNR) required for a successful decoding at UAVs. For quadrature phase shift keying (QPSK) modulation, the minimum transmit power of device $i$ needed to reach a bit error rate requirement of $\delta$ is:\vspace{-0.2cm}
 \begin{equation} \label{Pt_min}
P_t^{ij} = {\left[ {{Q^{ - 1}}\left( \delta  \right)} \right]^2}\frac{{{R_b}{N_o}}}{2}{10^{\eta /10}}{\left( {\frac{{4\pi {f_c}{d_{ij}}}}{c}} \right)^2},
\end{equation}
where $Q^{ - 1}(.)$ is the inverse $Q$-function, $N_o$ is the noise power spectral density, and $R_b$ is the transmission bit rate. Note that, to derive (\ref{Pt_min}) using (\ref{Pr}),  we use the bit error rate expression for QPSK modulation as ${P_e} = Q(\sqrt{\frac{{2P_r^{ij}}}{{{R_b}{N_o}}}})$ \cite{Goldsmith}.

Our first goal is to optimally move and deploy the UAVs in a way that the total transmit power of devices to reach the minimum SNR requirement for a successful decoding at the UAVs is minimized. In fact, the objective function is:\vspace{-0.2cm}
\begin{align} \label{opt1}
&\min \limits_{{\mathcal{C}_j},{\boldsymbol{\mu_j}}} \sum\limits_{j = 1}^K {\sum\limits_{i \in {\mathcal{C}_j}} {P_t^{ij}} }, \,\,\,j\in \mathcal{K},
\end{align}
where $P_t^{ij}$ is the transmit power of device $i$ to UAV $j$, and $K$ is the number of UAVs. Also, $\mathcal{C}_j$ is the set of devices assigned to UAV $j$, and $\boldsymbol{\mu_j}$ is the 3D location of UAV $j$.

From (\ref{Pt_min}), we can observe that the transmit power is directly proportional to the distance squared. Hence, minimizing the power is equivalent to minimizing the distance squared. Then, using (\ref{dmin}), (\ref{Pt_min}), and (\ref{opt1}), and considering the constraint on the maximum number of devices that can connect to each UAV, our optimization problem can be formulated as:\vspace{-0.2cm}
\begin{align} 
&\left\{ {\mathcal{C}_j^*,\boldsymbol{\mu_j^*}} \right\} = \mathop {\arg \min }\limits_{{\mathcal{C}_j},{\boldsymbol{\mu_j}}} \sum\limits_{j = 1}^K {\sum\limits_{i \in {\mathcal{C}_j}} {d_{ij}^2} } \,\,,\,\,j\in \mathcal{K}, \label{opt_main}\\
{\text{s.t.}}\,\,&{\mathcal{C}_j} \cap {\mathcal{C}_m} = \emptyset ,\,\,j \ne m,\,\,\,\, j,m\in \mathcal{K}, \label{cons1} \\
&\sum\limits_{j = 1}^K {|{\mathcal{C}_j}|}  = L, \label{cons2}\\\
&{d_{ij}} \le \frac{h_j}{{\sin \left( {P_{\text{LoS}}^{ - 1}(\varepsilon )} \right)}}, \\
&|{\mathcal{C}_j}| \le {M_j},
\end{align}
where $|\mathcal{C}_j|$ is the number of devices assigned to UAV $j$, $L$ is the total number of devices, and $M_j$ is the maximum number of devices that UAV $j$ can support. (\ref{cons1}) and (\ref{cons2}) guarantee that each device connects to only one UAV.

Clearly, we can consider the set of devices which are assigned to a UAV as a cluster, and place the corresponding UAV at the center of the cluster. Placing a UAV at the cluster center ensures that the UAV has a minimum total distance squared to all the cluster members. Hence, to solve problem (\ref{opt_main}), we need to find $L$ clusters, and their centers which effectively correspond to the locations of the UAVs. Note that, in a time varying network in which the location of IoT devices change, the clusters will also change. Consequently, the location of UAVs as the center of the clusters must be updated accordingly. However, moving these UAVs to the center of the clusters should be done with a minimum energy consumption. Therefore, in the mobile scenario, while finding the optimal paths of UAVs, we need to determine which UAV must go to which cluster center. Next, we present the IoT devices' clustering approach for minimizing the total transmit of the IoT devices.


\begin{figure}[!t]
	\begin{center}
		\vspace{-0.2cm}
		\includegraphics[width=8.45cm]{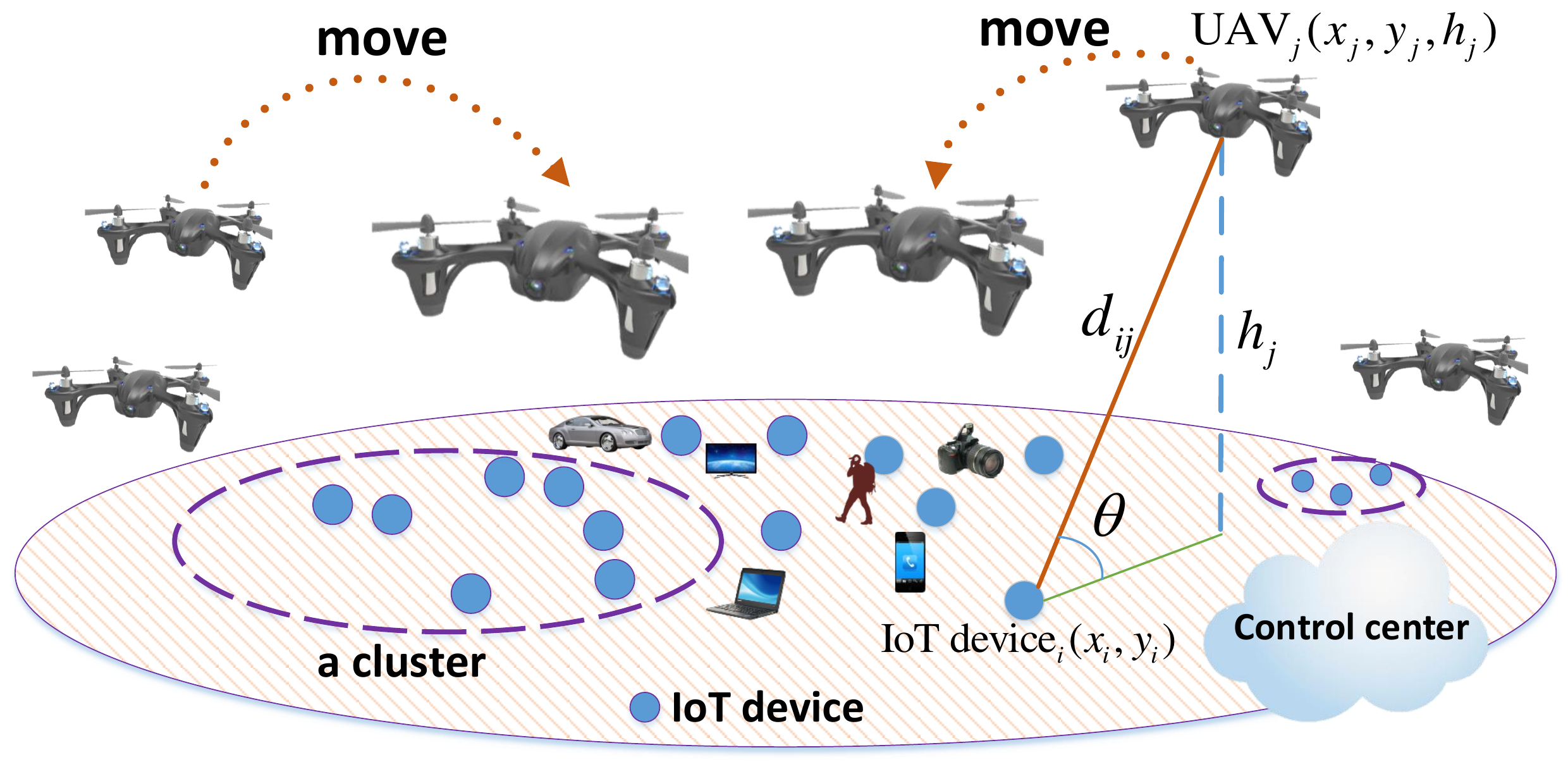}
		\vspace{-0.4cm}
		\caption{ \small System model.\vspace{-0.72cm}}
		\label{Nu}
\end{center}
\end{figure}

\section{Clustering IoT Devices}

Our first step is to optimally cluster the devices and deploy the UAVs at the center of the formed clusters so as to minimize the transmit power of the ground IoT devices. 

We solve the problem in (\ref{opt_main}) by exploiting the constrained $K$-mean clustering approach \cite{Hoppner}. In the $K$-mean clustering problem, given $L$ points in $\mathds{R}^n$, the goal is to partition the points into $K$ disjoint clusters such that the sum of distances squared of the points to their corresponding cluster center is minimized. Hence, considering (\ref{opt1}) and (\ref{opt_main}), the total transmit power of devices is minimized by placing the UAVs in the center of the optimal clusters. This problem can be solved in two steps using an iterative process. The first step is related to the assignment, and the second step is called update.

\subsection{Assignment Step}

In the assignment step, given the location of the clusters' center (given all $ \boldsymbol{\mu_j}$), we find the optimal clusters for which the total distance squared between the cluster members and their center is minimized. Therefore, in our problem, given the location of the UAVs, we will first determine the optimal assignment of the devices to UAVs which can be written as:\vspace{-0.1cm}
\begin{align} \label{assign}
&\min\limits_{A_{ij}} \sum\limits_{j = 1}^K {\sum\limits_{i = 1}^L {{A_{ij}}||{\boldsymbol{v_i}} - {\mu _j}|{|^2}} }, \\
{\rm{s}}{\rm{.t}}{\rm{.}}&\sum\limits_{i = 1}^L {{A_{ij}}}  \le {M_j}, \,\,\,j\in\mathcal{K},\\
&\sum\limits_{j = 1}^K {{A_{ij}}}  = 1, \,\,\,i\in\mathcal{L},\\
&{A_{ij}}||{\boldsymbol{v_i}} - {\boldsymbol{\mu_j}}|| \le \frac{h_j}{{\sin \left( {P_{\text{LoS}}^{ - 1}(\varepsilon )} \right)}},\\
&{A_{ij}} \in \{ 0,1\}, 
\end{align}
where $\boldsymbol{v_i}=(x_i,y_i)$ is the two-dimensional location vector of device $i$, and $A_{ij}$ is equal to 1 if device $i$ is assigned to UAV $j$, otherwise $A_{ij}$ will be equal to 0. The problem presented in (\ref{assign}) is an integer linear programming which is solved by using the cutting plane method \cite{garf}. 

\subsection {Update Step}

In the update step, given the clusters obtained in the assignment step, we update the location of the UAVs which is equivalent to updating the center of the clusters. Thus, the update location of UAVs is the solution to the following optimization problem:
\begin{align}\label{update}
&\mathop {\min }\limits_{({x_{u,j}},{y_{u,j}},h_j)} \sum\limits_{i \in {\mathcal{C}_j}} {{{({x_{u,j}} - {x_i})}^2} + {{({y_{u,j}} - {y_i})}^2}}+{h_j}^2 ,\\
&\text{s.t.}\,\,{({x_{u,j}} - {x_i})^2} + {({y_{u,j}} - {y_i})^2} + h_j^2\left( {1 - \frac{1}{{{{\sin }^2}\left( {P_{\text{LoS}}^{ - 1}(\varepsilon )} \right)}}} \right) \le 0, \label{consupdate} \nonumber\\
& \text{for all} \,\,\, i \in {\mathcal{C}_j},\,\, \text{and}\,\, j\in\mathcal{K}.
\end{align} 

\begin{theorem}
	\normalfont
The solution to (\ref{update}) is $\boldsymbol{{s^*}} = (x_{u,j}^*,y_{u,j}^*,h_j^*) =   - \boldsymbol{P{(\lambda )^{ - 1}}Q(\lambda )}$, with the vector $\boldsymbol{\lambda}$ that maximizes the following concave function:
\begin{align}
    &\mathop {\max {\rm{ }}}\limits_{\boldsymbol{\lambda}}  \frac{1}{2} \boldsymbol{Q{(\lambda )^T} P{(\lambda )^{ - 1}}Q(\lambda )} +   r(\boldsymbol{\lambda}),\\
	&\textnormal {s.t.} \,\,\,\boldsymbol{\lambda}  \ge 0,
\end{align}
where $\boldsymbol{P(\lambda)}= \boldsymbol{P_o} + \sum\limits_i {{\lambda _i}{P_i}}$, $\boldsymbol{Q(\lambda)}=\boldsymbol{Q_o} + \sum\limits_i {{\lambda _i}{Q_i}}$ and $r(\boldsymbol{\lambda})={r_o} + \sum\limits_i {{\lambda _i}{r_i}}$, with $\boldsymbol{P_o}$, $\boldsymbol{Q_o}$, $r_o$, $\boldsymbol{P_i}$, $\boldsymbol{Q_i}$, and $r_i$ given in the proof.	 	
\end{theorem}    
\begin{proof}
As we can see from (\ref{update}), the optimization problem is a quadratically constrained quadratic program (QCQP) whose general form is given by: 
\begin{align}\label{QCQP}
&\mathop {\min }\limits_{\boldsymbol{s}} \,{\rm{ }}\frac{1}{2}\boldsymbol{{s^T}{P_o}s} + \boldsymbol{Q_o^Ts} + {r_o},\\
&\text{s.t.}\,\,\frac{1}{2}\boldsymbol{{s^T}{P_i}s} + \boldsymbol{Q_i^Ts} + {r_i}\le 0,\,\,\,i = 1,...,|{\mathcal{C}_j}|.
\end{align}
Given (\ref{update}) and (\ref{consupdate}), we have:\vspace{0.2cm}

$\boldsymbol{{P_o}} = \left[ {\begin{array}{*{20}{c}}
	{2|{\mathcal{C}_j}|}&0&0\\
	0&{2|{\mathcal{C}_j}|}&0\\
	0&0&{2|{\mathcal{C}_j}|}
	\end{array}} \right]$, $\boldsymbol{{P_i}} = \left[ {\begin{array}{*{20}{c}}
	2&0&0\\
	0&2&0\\
	0&0&\omega 
	\end{array}} \right]$,\vspace{0.1cm}\\  $\omega  = 1 - \frac{1}{{{{\sin }^2}\left( {P_{\text{LoS}}^{ - 1}(\varepsilon )} \right)}}$, $\boldsymbol{Q_o}={\left[ {\begin{array}{*{20}{c}}
		{ - 2\sum\limits_{i = 1}^{|{\mathcal{C}_j}|} {{x_i}} }&{ - 2\sum\limits_{i = 1}^{|{\mathcal{C}_j}|} {{y_i}} }&0
		\end{array}} \right]^T}$,\vspace{0.2cm}\\ $\boldsymbol{{Q_i}} = {\left[ {\begin{array}{*{20}{c}}
		{ - 2{x_i}}&{ - 2{y_i}}&0
		\end{array}} \right]^T}$, ${r_o} = \sum\limits_{i = 1}^{|{\mathcal{C}_j}|} {x_i^2}  + \sum\limits_{i = 1}^{|{\mathcal{C}_j}|} {y_i^2}$, and \vspace{0.1cm}\\ ${r_i} = x_i^2 + y_i^2$. Note that, $\boldsymbol{P_o}$ and $\boldsymbol{P_i}$ are positive semidefinite matrices, and, hence, the QCQP problem in (\ref{QCQP}) is convex. Now, we write the Lagrange dual function as:\vspace{-0.1cm}
\begin{align}
f(\lambda ) = \mathop \text{inf}\limits_{\boldsymbol s} \biggl[&\frac{1}{2}\boldsymbol{{s^T}{P_o}s} + \boldsymbol{Q_o^Ts} + {r_o}\nonumber  \\
&+ \sum\limits_i {{\lambda _i}\left( {\frac{1}{2}\boldsymbol{{s^T}{P_i}s} + \boldsymbol{Q_i^Ts} + {r_i}} \right)}\biggr]\nonumber \\
&= \mathop \text{inf}\limits_{\boldsymbol s} \left[ {\frac{1}{2}\boldsymbol{{s^T}P(\lambda )s} + \boldsymbol{Q{{(\lambda )}^T}s} + r(\boldsymbol{\lambda} )} \right].\nonumber
\end{align}  
Clearly, by taking the gradient of the function inside the infimum with respect to $s$, we find $\boldsymbol{{s^*} =  - P{(\lambda )^{ - 1}}Q(\lambda )}$. As a result, using $\boldsymbol{{s^*}}$, $f(\boldsymbol{\lambda} ) = \frac{1}{2}\boldsymbol{Q{(\lambda )^T}P{(\lambda )^{ - 1}}Q(\lambda )} + {r(\boldsymbol{\lambda} )}$. Finally, the dual of problem (\ref{QCQP}) or (\ref{update}) will be:\vspace{-0.1cm}
\begin{align}
\text{max}\,\, f(\boldsymbol{\lambda}), \,\,\textnormal {s.t.} \,\,\,\boldsymbol{\lambda}  \ge 0,
\end{align}
which proves Theorem 1.
\end{proof}

The assignment and update steps are applied iteratively until there is no change in the  location update step. Clearly, at each iteration, the total transmit power is reduced and the objective function is monotonically decreasing. Hence, the solution converges after several iterations. 

In summary, for given locations of the ground IoT devices, we determined the optimal locations of the UAVs (cluster centers) for which the transmit power of the IoT devices used for reliable uplink communications is minimized. In a mobile IoT network, the UAVs must update their locations and follow the cluster centers as they evolve due to time-varying dynamics.  
Next, we investigate how to optimally move the UAVs to the center of the clusters.    

\section{Mobility of UAVs: Optimal Transport Theory}
Here, we find the optimal trajectory of the UAVs to guarantee the reliable uplink transmissions of mobile IoT devices. To move along the optimal trajectories, the UAVs must spend a minimum total energy on their mobility so as to remain operational for a longer time. In the considered mobile ground IoT network, the location of the devices and their availability might change, and hence, the clusters will change. Consequently, the UAVs must frequently update their locations accordingly. Now,  given the location of the cluster centers obtained in Section III, and initial locations of the UAVs, we determine which UAV should fly to which cluster center such that the total energy consumed for their mobility is minimized. 
In other words, given $\mathcal{I}$ and $\mathcal{J}$, the initial and new sets of UAVs' locations, one needs to find the optimal mapping between these two sets in a way that the energy used for transportations (between two sets) is minimized.

This problem can be modeled using discrete \emph{optimal transport theory} \cite{villani2003}. In its general form, optimal transport theory deals with finding an optimal transportation plan between two sets of points that leads to a minimum transportation cost \cite{villani2003}. These sets can be either discrete or continuous, with arbitrary distributions, and can have a general transportation cost function. The optimal transport theory was originated from the following Monge problem \cite{villani2003}. Given piles of sands and holes with the same volume, find the best move (transport map) to completely fill up the holes with the minimum total transportation cost. In general, this problem does not necessarily have a solution  as each point must be mapped to only one location. However, Kantorovich relaxed this problem by using transport plans instead of maps, in which one point can go to multiple points \cite{villani2003}.  

 In our model, the UAVs need to move from initial locations to the new destinations. The transportation cost for each move is the energy used by each UAV for the mobility. We model this problem based on the discrete Monge-Kantorovich problem as follows \cite{xia}: \vspace{-0.10cm}
\begin{align} \label{transport1}
&\min\limits_{Z_{kl}} \sum\limits_{l \in \mathcal{J}} {\sum\limits_{k \in \mathcal{I}} {{E_{kl}}{Z_{kl}}} }, \\
\text{s.t.}\,&\sum\limits_{l \in \mathcal{J}} {{Z_{kl}}}  = {m_k},\\
&\sum\limits_{k \in \mathcal{I}} {{Z_{kl}}}  = {m_l},\\
&{Z_{kl}} \in \{ 0,1\},
\end{align}
where $\mathcal{I}$ and $\mathcal{J}$, are the initial and new sets of UAVs' locations. $\boldsymbol{Z}$ is the $\mathcal{|J|\times|I|}$ transportation plan matrix with each element $Z_{kl}$ being 1 if UAV $k$ is assigned to location $l$, and 0 otherwise. $E_{kl}$ is the energy used for moving a UAV from its initial location with index $k \in \mathcal{I}$ to a new location with index $l \in \mathcal{J}$. $m_l$ and $m_k$ are the number of points (UAVs) at the locations with indices $l$ and $k$. 
The energy consumption of a UAV moving with a constant speed as a function of distance is given by \cite{di}:\vspace{-0.3cm}
\begin{equation} \label{energy}
E(D,v) = \int\limits_{t = 0}^{t = D/v} {p(v)dt}  = \frac{{p(v)}}{v}D,  
\end{equation}
where $D$ is the travel distance of the UAV, $v$ is the constant speed, $t$ is the travel time, and $p(v)$ is the power consumption as a function of speed. As we can see from (\ref{energy}), energy consumption for mobility is linearly proportional to the travel distance. Using the Kantorovich-duality, the discrete optimal transport problem in (\ref{transport1}) is equivalent to:\vspace{-0.1cm}
\begin{align}\label{dual}
&\max\limits_{\varphi,\xi} \left[ {\sum\limits_{k \in \mathcal{I}} {\xi (k)}  - \sum\limits_{l \in \mathcal{J}} {\varphi (l)} } \right],\\
& \text{s.t.}\,\, \varphi (l)-\xi (k)  \le {E_{kl}},
\end{align}
where $\xi :\mathcal{I} \to \mathds{R}$, and $\varphi :\mathcal{J} \to \mathds{R}$ are unknown functions of the maximization problem. The dual problem in (\ref{dual}) is used to solve the primal problem in (\ref{transport1}) by applying the complementary slackness theorem \cite{villani2003}. In this case, the optimal solution including the optimal transport plan between $\mathcal{I}$ and $\mathcal{J}$ is achieved when $\varphi (l)-\xi (k) = {E_{kl}}$ \cite{villani2003}. Here, to find the optimal mapping between initial set of locations and the destination set, we use the revised simplex method \cite{ford}. The result will be the transportation plan ($\boldsymbol{Z}$) that optimally assigns the UAVs to the destinations. Hence, the location of the UAVs are updated according to the new destinations. Subsequently, having the destinations of each UAV at different time instances, we can find the optimal trajectory of the UAVs. As a result, given the optimal paths, the UAVs are able to serve the mobile IoT devices in an energy-efficient way.

 

\section{Simulation Results and Analysis}
In our simulations, the IoT devices are deployed within a geographical area of size $1.2\,\text{km}\times 1.2 \,\text{km}$. We consider UAV-based communications in an urban environment with $\psi=11.95$, and $\beta=0.14$ at 2\,GHz carrier frequency \cite{HouraniModeling}. Moreover, we use the energy consumption model for UAVs' mobility as $E(D,v) = D\left( {0.95{v^2} - 20.4v + 130} \right)$ \cite{di}. Furthermore, in a time-varying network, to capture the mobility and availability of the ground IoT devices,   
we generate the new devices' locations by adding zero mean Gaussian random variables to the initial devices' locations. Table I lists the simulation parameters. Note that,
all statistical results are averaged over a large number of independent runs. 
\begin{table}[!t]
	\normalsize
	\begin{center}
		\caption{\small Simulation parameters.}
		\label{TableP}
		\resizebox{7.2cm}{!}{
			\begin{tabular}{|c|c|c|}
				\hline
				\textbf{Parameter} & \textbf{Description} & \textbf{Value} \\ \hline \hline
				$f_c$	&     Carrier frequency     &      2\,GHz     \\ \hline 
				$v$	&     UAV speed     &      10\,m/s    \\ \hline
				$\delta$	&    Bit error rate requirement      &      $10^{-8}$     \\ \hline
				$\epsilon$	&     $P_\text{LoS}$ requirement      &        0.95   \\ \hline
				$N_o$	&     Noise power spectral density    &        -170\,dBm/Hz   \\ \hline
				$R_b$	&     Transmission data rate      &       200\,Kbps \\ \hline
				$B$	&     Transmission bandwidth per device   &    200\,KHz \\ \hline
				$\eta$	&     Additional path loss to free space      &   5\,dB \\ \hline
				
			\end{tabular}}
			
		\end{center}\vspace{-0.3cm}
	\end{table}

Figure \ref{cluster} shows a snapshot of the optimal devices' clustering as well as the optimal UAVs' locations resulting from our proposed approach. In this example, 5 UAVs are used to support 100 IoT devices. We assume that each UAV has a limited number of resource blocks which can be allocated to at most 30 devices. Therefore, we have 5 clusters with a maximum size of 30 devices and 5 cluster centers corresponding to the locations of the UAVs.  As we can see from Figure \ref{cluster}, the location of IoT devices significantly impacts the number of devices per cluster and also the optimal locations of the UAVs. In this figure, the minimum and maximum cluster sizes are 3 and 27.

\begin{figure}[!t]
	\begin{center}
		\vspace{-0.2cm}
		\includegraphics[width=7.6cm]{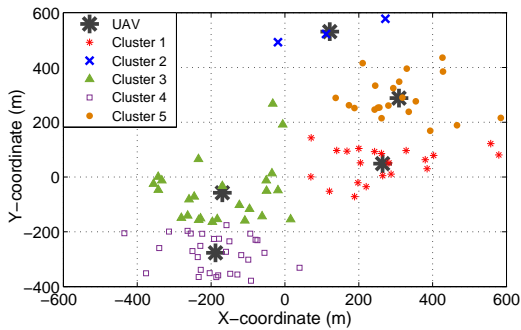}
		\vspace{-0.4cm}
		\caption{ \small Optimal clusters and UAVs' locations in one snapshot.\vspace{-0.7cm}}
		\label{cluster}
	\end{center}
\end{figure}

%
%

Figure \ref{Pt_voronoi} shows the total transmit power of devices versus the number of UAVs averaged over multiple simulation runs. In this figure, the performance of the proposed approach is compared with the fixed Voronoi case which is known to be a typical deployment method for static base stations. Note that, for a fair comparison, we assume that the total number of resource blocks is fixed ($L$), and hence, the number of resources per UAV is $\left\lceil {\frac{L}{K}} \right\rceil$. In other words, the maximum size of each cluster will decrease as the number of UAVs increases.   In the Voronoi method, assuming a uniform distribution of devices, we fix the location of UAVs at an altitude of 500\,m, and then, we assign each device to the closest UAV. However, in the proposed clustering algorithm, we find the optimal clusters and deploy the UAVs at the center of the clusters. As shown in Figure \ref{Pt_voronoi}, the proposed method outperforms the classical Voronoi scheme as the UAVs can be placed closer to the devices. As expected,  increasing the number of UAVs reduces the total transmit power of IoT devices. For instance, when the number of UAVs increases from 4 to 8, the total transmit power decreases from 77\,mW to 38\,mW for the proposed method, and from 115\,mW to 95\,mW for the Voronoi case. Figure \ref{Pt_voronoi} shows that our approach results in about 56\% reduction in the transmit power of the IoT devices.

\begin{figure}[!t]
	\begin{center}
		\vspace{-0.2cm}
		\includegraphics[width=6.8cm]{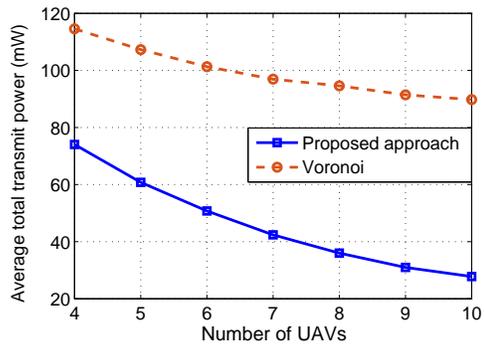}
		\vspace{-0.2cm}
		\caption{ \small Average of total transmit power vs. number of UAVs.\vspace{-0.4cm}}
		\label{Pt_voronoi}
	\end{center}
\end{figure}

Figure \ref{Trajectory} shows the trajectory of one of the UAVs in a mobile IoT scenario derived from optimal transport theory. 
Here, we consider 8 UAVs, and 400 devices whose locations are updated at each time by adding a Gaussian random variable with  $N(0,50\,\text{m})$ to the previous locations. 
 Clearly, since the locations of the devices may change over time, the optimal clusters must be updated accordingly. In Figure \ref{Trajectory}, the red dots correspond to the optimal destinations of the UAV at different times. In fact, as the clusters are changing over time, the UAV uses the proposed scheme to optimally move to one of the new cluster centers. 

\begin{figure}[!t]
	\begin{center}
		\vspace{-0.2cm}
		\includegraphics[width=7.2cm]{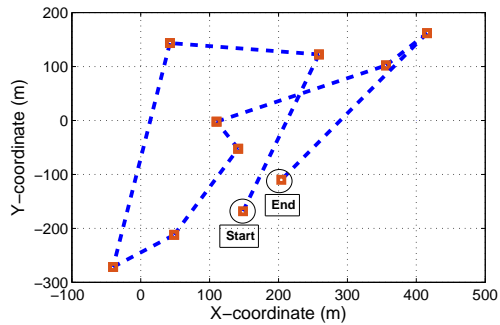}
		\vspace{-0.5cm}
		\caption{ \small Trajectory of a UAV in a mobile IoT network.\vspace{-0.8cm}}
		\label{Trajectory}
	\end{center}
\end{figure}

Figure \ref{EnergyConsumption} shows the energy consumed by each UAV during its mobility. In this case, we use 8 UAVs for supporting 400 ground IoT devices. We consider the network at 10 time instances during which the UAVs move at a speed of 10\,\text{m/s} while updating their locations. As shown in Figure \ref{EnergyConsumption}, for the given scenario, the total amount of energy that UAVs use for mobility is around 106\,\text{kJ}. Note that, this is the minimum total energy consumption that can be achieved via the optimal transport of the UAVs. As shown, different UAVs spend different amount of energy on the mobility. Depending on the optimal clustering of devices over time, different UAVs might have different travel distances to the cluster centers. For instance, UAV 1 consumes 1.8 times more energy than UAV 3. Hence, the number of UAVs may also change over time.
   
\begin{figure}[!t]
	\begin{center}
		\vspace{-0.5cm}
		\includegraphics[width=6.8cm]{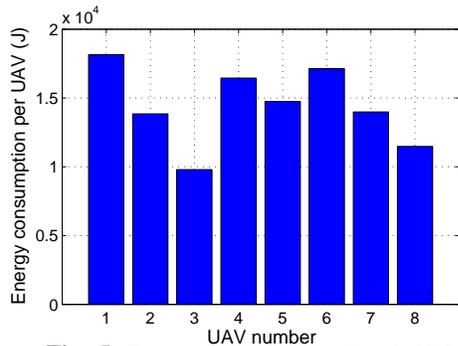}
		\vspace{-0.4cm}
		\caption{ \small Energy consumption of each UAV.\vspace{-0.15cm}}
		\label{EnergyConsumption}
	\end{center}
\end{figure}

Figure \ref{UAV_loss3} shows the energy consumption per UAV when the number of UAVs changes. Here, we assume that, initially the UAVs are optimally deployed for a given IoT system, however, after a while some of the UAVs ($q$ UAVs) will not be operational due to the lack of battery. Consequently, the number of UAVs decreases and the remaining UAVs must update their locations to maintain the power efficiency of the ground devices. 
In Figure \ref{UAV_loss3},  for the average case, we take the average of energy over all possible combinations of removing $q$ UAVs among the total UAVs. However, in the worst-case scenario, we remove the $q$ UAVs whose loss leads to the highest energy consumption for the remaining UAVs. 
 Clearly, as more UAVs become inoperational, the energy consumption of the functioning UAVs will increase. For example, when the number of lost UAVs increases from 2 to 4, the average energy consumption per UAV increases from 1520\,J to 2510\,J.\vspace{0.09cm} 

\begin{figure}[!t]
	\begin{center}
		\vspace{-0.5cm}
		\includegraphics[width=6.6cm]{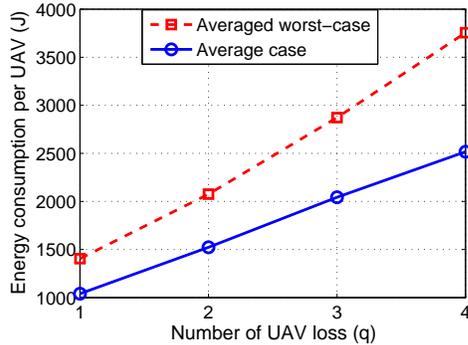}
		\vspace{-0.70cm}
		\caption{ \small Energy consumption vs. number of battery depleted UAVs.\vspace{-0.76cm}}
		\label{UAV_loss3}
	\end{center}
\end{figure}

\section{conclusions}\vspace{-0.007cm}
In this paper, we have proposed a novel framework for efficiently deploying and moving UAVs to collect data from ground IoT devices. In particular, we have determined the optimal clustering of IoT devices as well as the optimal deployment and mobility of the UAVs such that the total transmit power of IoT devices is minimized while meeting a required bit error rate. To perform clustering given the limited capacity of each UAV, we have adopted the constrained size clustering approach. 
Furthermore, we have obtained the optimal trajectories that are used by the UAVs to serve the mobile IoT devices with a minimum energy consumption.
 The results have shown that by carefully clustering the devices  and deploying the UAVs, the total transmit power of devices significantly decreases compared to the classical Voronoi-based deployment. Moreover, we have shown that, by intelligently moving the UAVs, they can remain operational for a longer time while serving the ground devices. \vspace{-0.01cm}   
\def\baselinestretch{1.00}
\bibliographystyle{IEEEtran}
\vspace{-0.01cm}
\bibliography{references}
\vspace{-0.5cm}
\end{document}